\documentclass[12pt]{amsart}
\usepackage{amsthm,amsmath,amsxtra,amscd,amssymb,xypic,color}
\numberwithin{equation}{section}

\def\beq{\begin{equation}}
\def\eeq{\end{equation}}
\def\bea{\begin{eqnarray}}
\def\eea{\end{eqnarray}}
\def\p{\partial}
\def\G{\Gamma}
\def\g{\gamma}

\def\a{\alpha}

\def\res{{\rm res}}

\def\wh{\widehat}

\def \matrix #1 {\left(\begin{array}{cc} #1 \end{array}\right)}

\newtheorem{theo}{Theorem}[section]

\newtheorem{cor}[theo]{Corollary}
\newtheorem{lem}[theo]{Lemma}

\theoremstyle{definition}

\newtheorem{rem}[theo]{Remark}

\begin{document}
\title{Triangular reductions of $2D$ Toda hierarchy}
\author{Igor Krichever \and Anna Ilyina}\thanks{This work has been funded by the Russian Academic Excellence Project
'5-100'}
\address{Skolkovo Institute of Science and Technology, Moscow; Columbia University, New York;
National Research University Higher School of Economics, Russian Federation}
\email{krichev@math.columbia.edu}

\address{Skolkovo Institute of Science and Technology; National Research University Higher School of Economics, Russian Federation}
\email{ekrez@yandex.ru}
\begin{abstract} New reductions of the 2D Toda equations associated with low-triangular difference operators are proposed. Their explicit Hamiltonian description is obtained.
\end{abstract}
\maketitle

\section{Introduction}

A recent burst of interest to a theory of linear difference operators has been motivated by their
connection with the theory of discrete-time integrable systems (pentagram map and it's higher dimensional generalization), representation theory (Coxeter's friezes) and the theory of cluster algebras.

The pentagram map is defined for $n$-gons in $\mathbb{RP}^{2}$ as follows: a vertex $v_{i}$ of $n$-gon $(v_{1},$ $\ldots,$ $v_{n})$ is mapped to a point which is the intersection of two
diagonals  $(v_{i-1}, v_{i+1}) $ and $(v_{i}, v_{i+2}).$ If  $n$ and $k+1$ are co-prime, then, as shown in \cite{ovstab}, the moduli space of $n$-gons
in $\mathbb{RP}^{k}$ is isomorphic, as algebraic varieties, to the space $\mathcal{E}_{k+1,n}$ of $n$-periodic linear difference equations
\beq\label{difeq}
V_i=a_i^{(1)}V_{i-1}-a_i^{(2)}V_{i-2}+\cdots+(-1)^{k-1} a_i^{(k)}V_{i-k}+(-1)^{k}V_{i-k-1},\
\eeq
whose {\it all solutions} are (anti)periodic
\beq\label{superperiodic}
V_{i+n}=(-1)^{k}V_{i}.
\eeq
In \cite{Kr1}  such equations were called {\it superperiodic}.

More generally, equations (\ref{difeq}) without constraint (\ref{superperiodic}) correspond to, the so-called, twisted $n$-gon in $\mathbb{RP}^k$, that is a sequence  $v_j\in \mathbb{RP}^k, j\in \mathbb Z,$ for which there is a projective linear transformation $M$ of $\mathbb {RP}^{k}$ such that $v_{j+n}=Mv_j$.

\medskip
In \cite{ovstab} it was shown that the pentagram map is a discrete complete integrable system, i.e. that the space of $n$-periodic lower-triangular operators (\ref{difeq}) of order 3 is a Poisson manifold and a complete set of integrals in involution for the pentagram map was constructed. Algebraic-geometrical integrability of the pentagram map was proved in \cite{soloviev}.

In \cite{ovstab2} an explicit construction of duality of the spaces $\mathcal{E}_{k+1,n}$ and $\mathcal{E}_{n-k-1,n}$ was proposed, which is a generalization of  the classical Gale duality for $n$-gons. In \cite{Kr1} this duality was clarified and connected with the theory of commuting difference operators. Along the way of establishing this connection a spectral theory of strictly lower triangular difference operators
\beq\label{operator}
L=T^{-k-1}+\sum_{j=1}^{k}a_{i}^{(j)}T^{-j},\ a_{i}^{(j)}=a_{i+n}^{(j)}
\eeq
was developed. Here $T$ is the shift operator, $T\psi_{j}=\psi_{j+1}.$ Throughout the paper it is assumed that the leading coefficient of $L$ is non-zero:
\beq\label{a1}
a_j^{(1)}\neq 0\,.
\eeq
The spectral theory of the triangular difference operators is of its own interest. The starting point of this work was an observation: the spectral theory of triangular operators is naturally connected with a special reduction of
the $2D$ Toda hierarchy.
\begin{rem} For definiteness, in this paper we consider only the case of lower-triangular reductions, since the involution
$L\to L^*$, where
\beq\label{adj}
L^*=T^{k+1}+\sum_{j=1}^k T^ja_i^{(j)}=T^{k+1}+\sum_{j=1}^k a_{i+j}^{(j)}T^j.
\eeq
is the formal adjoint operator, establishes an equivalence of the cases of lower- and upper-triangular operators.
\end{rem}

Recall, that the $2D$ Toda equations by themselves
\begin{eqnarray}\label{2dtoda}
\partial^2_{\xi\eta}\varphi_{i}=e^{\varphi_{i}-\varphi_{i+1}}-e^{\varphi_{i-1}-\varphi_{i}}
\end{eqnarray}
is the compatibility condition of two linear problems
\begin{equation}\label{2toda}
 \begin{cases}
\partial_{\xi}\Psi_{i}=v_{i}\Psi_{i}+\Psi_{i-1}\\
\partial_{\eta}\Psi_{i}=c_{i}\Psi_{i+1}, \ c_{i}=e^{\varphi_{i}-\varphi_{i+1}}
 \end{cases}
\end{equation}
The $2D$ Toda hierarchy  is a system of commuting flows on a space of functions $\varphi$ depending on one discrete variable $i$ and two sets of continuous variables $t_m^{\pm}$, i.e. $\varphi=\varphi_i(t_1^+,t_1^-,t_2^+,t_2^-,\ldots)$.
The variables $t_{1}^{+}$ and $t_{1}^{-}$ are identified with $\xi$ and $\eta$.
The hierarchy is the compatibility condition of a system of the linear problems
\beq\label{hierarch}
\p_{t_{m}^{\pm}}\Psi=L^{\pm}_{m}\Psi
\eeq
where $L^{\pm}_{m}$ are difference operators of the form:
\beq\label{LJ}
L_{m}^\pm=\sum_{j=0}^{m} a_{i,m}^{(j,\pm)}T^{\pm j}
\eeq
with the leading coefficients
\beq\label{leading}
\ a_{i,m}^{(m,-)}=1,\ \ \ \ \ \ \ a_{i,m}^{(m,+)}=e^{\varphi_i-\varphi_{i+m}}
\eeq
It is easy to check that the compatibility of the second equation in (\ref{2toda}) with (\ref{hierarch}) implies
\beq\label{LJ1}
a_{i,m}^{(0,-)}=\p_{t_m^-}\varphi_i ,\ \ \ \ \ \ a_{i,m}^{(0,+)}=0.
\eeq

\medskip
\noindent
{\bf Remark.} It is necessary to emphasize that although the hierarchy of any soliton equation, as a linear space of commuting vector fields is well-defined, in general there is no canonical choice of the basic "times" of the hierarchy or equivalently a canonical basis of commuting vector fields. The condition above that the operators $L_m^{\pm}$ are upper (lower) triangular operators of order $m$ fixes that ambiguity only partially. By this constraint the times are defined up to a linear triangular transformations $\tilde t_m^{\pm}=t_m^{\pm}+\sum_{\mu<m}c_\mu^{\pm}t_\mu^\pm$. We will comment more on that in Sections 2 and 3 below.

\medskip

Let us fix one of the times of the hierarchy: $t_{k+1}^-$  (or more generally a linear combination of the first $(k+1)$ times),
and consider solutions the hierarchy that {\it do not} depend on it, i.e.
\beq\label{stat}
\partial_{t_{k+1}^{-}}\varphi_{i}=0
\eeq
The space of such solutions can be identified with the space of the auxiliary operators $L_{k+1}^-$.
Note, that from (\ref{LJ1}) it follows that under the constraint (\ref{stat}) the operator $L=L_{k+1}^-$ becomes  strictly low triangular, i.e. it takes the form (\ref{operator}).

The restriction of each $t_m^{\pm}$ flow of the hierarchy onto the space of stationary with respect to $t_{k+1}^-$ solutions can be seen as a finite-dimensional system admitting Lax representation
\beq\label{lax}
\p_{t_m^{\pm}} L=[L_{m}^{\pm},L]
\eeq
For $\xi=t_1^+$ the auxiliary operator has the form $L_1^-=v_i+T^{-1}$ with $v_i=\p_{\xi} \varphi_i$ and  (\ref{lax}) is equivalent to the system equations for $a_i^{(1)}=e^{\varphi_i-\varphi_{i-1}}$ and $a_i^{(j)}, j=2\ldots,k$:
\begin{equation}\label{LT1}
\begin{cases}
\partial_{\xi} a_{i}^{(j)}=a_{i-1}^{(j-1)}-a_{i}^{(j-1)}+a_{i}^{(j)}(v_{i}-v_{i-j}), \ j=2,\ldots,k\\
 0=a_{i-1}^{(k)}-a_{i}^{(k)}+(v_{i}-v_{i-k-1}),\ \ v_i=\p_{\xi}\varphi_i.
\end{cases}
\end{equation}
Similarly, for $\eta=t_1^+$, we get the system
\begin{equation}\label{LT2}
\partial_{\eta} a_{i}^{(j)}=c_{i}a_{i+1}^{(j+1)}-c_{i-j-1}a_{i}^{(j+1)},  j=1,\ldots, k\\
\end{equation}
where $a_i^{(1)}=e^{\varphi_i-\varphi_{i-1}},\ c_i=e^{\varphi_i-\varphi_{i+1}}$.

The main goal of this paper is to construct a bi-Hamiltonian theory of systems (\ref{LT1}) and ({\ref{LT2}). We show that the space of strictly low-diagonal difference operators $L$ admits two different structures of the Poisson manifold and identify the corresponding Hamiltonians.

For $k=1$ the systems (\ref{LT1}) and (\ref{LT2}) have the most simple and interesting form:
\beq\label{lttoda1}
\p_{\xi}\varphi_{i-1}-\p_\xi\varphi_{i+1}=e^{\varphi_{i}-\varphi_{i-1}}-e^{\varphi_{i+1}-\varphi_{i}}
\eeq
\beq\label{lttoda2}
\p_{\eta}\varphi_i-\p_\eta\varphi_{i-1}=e^{\varphi_{i-1}-\varphi_{i+1}}-e^{\varphi_{i-2}-\varphi_{i}}
\eeq
{\it A posteriori}, in these cases one of our main result can be directly verified. Namely, it is easy to check that the systems (\ref{lttoda1}) and (\ref{lttoda2}) are Hamiltonian with respect to the form $\omega=\sum_{i=1}^n d\varphi_i\wedge d\varphi_{i+1}, \varphi_i=\varphi_{i+n},$ with the Hamiltonians
\beq\label{ham}
H^-=\sum_{i=1}^n e^{\varphi_i-\varphi_{i-1}},\ \ H^+=\sum_{i=1}^n e^{\varphi_{i-2}-\varphi_i}\,, \varphi_i=\varphi_{i+n},
\eeq
respectively.
In this simple case the second Hamiltonian structure of equations (\ref{lttoda1}) and (\ref{lttoda2}) is not so obvious.
In the last  part we prove that under (one-to-one for odd $n$) change of variables
$e^{\varphi_i-\varphi_{i-1}}=x_i-x_{i-2}+e_1$ equations  (\ref{lttoda1}) take the form of Hamiltonian equations with respect to the form
$\widetilde \omega=\sum_{i=1}^n dx_i\wedge dx_{i-1}, x_i=x_{i+n}$, with the Hamiltonian
$$
\widetilde H^-=\sum_{i=1}^n  x_i^2(x_{i-1}-x_{i+1})
$$

\section{Necessary facts}

In this section we present necessary facts from the spectral theory of strictly lower-triangular operators and the construction of algebraic-geometrical solutions of the $2D$ Toda hierarchy.

\subsection{The spectral theory of lower-triangular difference operators.}

In the modern approach to the spectral theory of periodic difference operators central is the notion  of a {\it spectral curve} associated with each $n$-periodic difference operator $L$. By definition, points of the spectral curve parameterize Bloch solutions of the equation
\beq\label{eigen}
L\psi=E\psi,
\eeq
i.e. the solutions of (\ref{eigen}) that are eigenfunctions for the monodromy operator $T^{-n}$
\beq\label{bloch}
T^{-n}\psi=w\psi.
\eeq
Let $\mathcal{L}(E)$ be a space of the solutions of equation (\ref{eigen}). It is a linear space of dimension equal to the order of $L$. The monodromy operator preserves $\mathcal{L}(E)$ and, hence, defines on it a finite-dimensional operator $T^{-n}(E)$. The pairs of complex numbers $(w,E)$ for which there exists a common solution of equations (\ref{eigen}) and (\ref{bloch}) satisfy the characteristic equation:
$$R(w,E)=det\bigl( w\cdot 1-T^{-n}(E)\bigr)=0$$
Similarly, the same polynomial $R(w,E)$ can be obtained as the characteristic polynomial of the finite dimensional operator $L(w)$ that is a restriction of $L$ on the space $ \mathcal{T} (w):=\{\psi \ | \ w\psi_{i+n}=\psi_{i}\}$:
\beq\label{specR}
R(w,E)=det\bigl( E\cdot 1-L(w)\bigr)=0,\ \ L(w):=L \big|_{\mathcal{T}(w)}
\eeq

\medskip
The family of algebraic curves that arises as spectral curves depends on a family of operators. In case of   strictly lower-triangular difference operators in (\cite{Kr1}) it was noticed that the characteristic equation of such operator is of the form
\beq\label{spectralcurve}
R(w,E)=w^{k+1}-E^{n}+\sum_{i>0, j\ge 0,\ ni+(k+1)j<n(k+1)} r_{ij}w^{i}E^{j}=0,
\eeq
where $r_{1,0}=\prod_{i=1}^n a_{i}^{1}\neq 0$ (due to (\ref{a1})).

If $n$ and $k+1$ are co-prime, then the affine curve defined in $\mathbb C^2$ by (\ref{spectralcurve})
is compactified by one point $p_-$, where the functions $w(p)$ and $E(p)$, naturally defined
on $\Gamma$,  have pole of order $n$ and $k+1$, respectively. In other words, if one chooses a local coordinate $z$ in the neighborhood of $p_-$ such that $w=z^{-n}$, then the Laurent expansion of $E$ has the form:
\beq\label{Eexpansion}
E= z^{-k-1}\left(1+\sum_{s=1}^\infty e_{s}z^{s}\right), \ \ w=z^{-n}.
\eeq

As it was emphasized in \cite{Kr1}, the specific form of equation (\ref{spectralcurve}) allows one
to single out another marked point $p_+$ on $\Gamma$. Namely, it is the only preimage of $E=0$ where $w=0$. It turns out that at this point $E=E(p)$ has a simple zero, and the function $w=w(p)$ has zero of order $n$
\beq\label{w0}
w=\frac{1}{r_{1,\,0}}E^{n}\left(1+\sum_{s=1}^\infty w_{s}E^{s}\right)
\eeq
Analytic properties of the Bloch solutions in the neighborhoods of marked points are described by the following two statements:
\begin{lem}\label{L1}(\cite{Kr1})
Let $L$ be an  operator of form (\ref{operator}) whose order and the period are co-prime. Then there is a unique formal series $E(z)$ of the form (\ref{Eexpansion}) such that the equation $L\psi=E\psi$ has a unique formal solution of the form
\beq\label{formalpsi}
\psi_{i}(z)=z^{i}\bigl(1+\sum_{s=1}^{\infty} \xi_{s}^{-}(i)z^{s}\bigr).
\eeq
with periodic coefficients $\xi_{s}^-(i)=\xi_{s}^-(i+n)$ normalized by the condition
$\xi_{s}^-(0)=0$.
\end{lem}
For further use let us briefly outline the proof.

\noindent {\it Proof.} The substitution of (\ref{formalpsi}) and (\ref{Eexpansion}) into the equation $L\psi=E\psi$ gives a system of difference equations for unknown constants $e_s$ and unknown functions $\xi_s(i)$ of the discrete variable $i$.
The first of them is the equation
\beq\label{xi1}
e_1+\xi_1^-(i)-\xi_1^-(i-k-1)=a_i^{(k)}\,.
\eeq
The periodicity constraint for $\xi_1^-$ uniquely defines
\beq\label{e1}
e_1=n^{-1}\sum_{i=1}^na_i^{(k)}
\eeq
and reduces difference equation (\ref{xi1}) of order $k+1$ to the difference equation of order 1:
\beq\label{xi11}
me_s+\xi_1^-(i)-\xi_1^-(i-1)=\sum_{j=0}^{m-1} a_{i-j(k+1)}^{(k)}\,,
\eeq
where $m$ is the integer $1\leq m<n$ such that $m(k+1)=1 \, ({\rm mod}\ n)$. Equation (\ref{xi11}) and the initial condition $\xi_1^-(0)=0$ uniquely defines $\xi_1^-(i)$.

For arbitrary $s$ the defining equation for $e_s$ and $\xi_s^-$ has the form:
\beq\label{xi}
e_s+\xi_s^-(i)-\xi_s^-(i-k-1)=Q_s(e_1,\ldots, e_{s-1};\xi_1,\ldots,\xi_{s-1}, a_i^{(j)})
\eeq
where $Q_s$ is an explicit function linear in $e_{s'}, \xi_{s'},\ s'<s,$ and polynomial in $a_i^{(j)}$.
The same arguments as above show that it has a unique periodic solution. The lemma is proved.
\medskip
\begin{lem}\label{L2} (\cite{Kr1})
Equation $L\psi=E\psi$ has an unique formal solution of the form
\beq\label{two}
\psi_{i}(E)=e^{\varphi_{i}}E^{-i}\bigl(1+\sum_{s=1}^{\infty}\xi_{s}^+(i)E^{s}\bigr), \ a_{i}^{(1)}=e^{\varphi_{i}-\varphi_{i-1}},
\eeq
normalized by the condition $\xi_{s}^+(0)=0$ .
\end{lem}
{\it Proof.} The substitution of (\ref{two}) into (\ref{eigen}) gives a system of nonhomogeneous first order difference equations equations for unknown coefficients $\xi_s^-$. For $s=1$ we have
\beq\label{chi1}
\xi_1^+(i)-\xi_1^+(i-1)=e^{\varphi_{i-2}-\varphi_{i}} a_i^{(2)}
\eeq
For any $s$ the equations have similar form
\beq\label{chi}
\xi_s^+(i)-\xi_s^+(i-1)=e^{-\varphi_i}q_s(\xi_1^+,\ldots,\xi^+_{s-1},a_i^{(j)})\,,
\eeq
which together with the initial conditions recurrently define $\xi_s^+(i)$ for all $i$.

The uniqueness of the formal solution (\ref{two}) implies
\begin{cor}\label{cor:1} The formal series (\ref{two}) is the Bloch solution, i.e. it satisfies (\ref{bloch}) with
\beq\label{wzero}
w(E)=\psi_{-n}(E)=r_{1,0}^{-1}E^n\left(1+\sum_{s=1}^\infty w_s E^s\right)
\eeq
\end{cor}

From Lemma \ref{L1} it follows that the component $\psi_i(p),\, p:=(w,E)\in \Gamma,$ of the Bloch solution $\psi(p)$ considered as a function on the spectral curve has {\it zero of order $i$} at the marked point $p_{-}$. Lemma \ref{L2} implies that $\psi_i(p)$ has {\it pole of order $i$} at the marked point $p_+$. Then the standard arguments (for details see \cite{babelon}) show that $\psi_i$ is a meromorphic function on $\Gamma$ having away of the marked point $g$ poles $\gamma_1,\ldots,\gamma_g$ that {\it do not depend} on $i$. These analytic properties are the defining properties of, the so-called, discrete Baker-Akhiezer (BA) function, introduced in \cite{kr3}. That establishes a connection of the spectral theory of the lower-triangular operators with the theory of commuting difference operators (see details in \cite{kr-diff}), and the theory of algebraic-geometrical solutions of the $2D$ Toda hierarchy.

\medskip

The correspondence
\beq\label{correspondence}
L \longmapsto \{\Gamma, D=\gamma_1+\cdots+\gamma_g\}
\eeq
where $\Gamma$ is the spectral curve of the operator $L$ and $D$ is the divisor of poles of the Bloch solution $\psi$ is usually referred  as {\it the direct spectral transform}. It is one-to-one correspondence of the open everywhere dense subsets of the spaces of operators and algebraic-geometrical spectral data. The construction of {\it the inverse} spectral transform is a particular case of the general construction of algebraic-geometrical solutions of the $2D$ Toda hierarchy which we now present.

\subsection{Algebraic-geometrical solutions of the $2D$ Toda lattice hierarchy}

Let $\G$ be a smooth genus $g$ algebraic curve with fixed local coordinates $z_\pm$ in the neighborhoods of two marked points $p_\pm \in \G,\ z_\pm(p_\pm)=0$, and let $t=\{t^\pm_j, \, j=1,2\ldots\}$ be a set of complex variables (it is assumed that only {\it finite} number of them are non-zero). Then, as shown in \cite{kr-2dtoda}:
\begin{lem}\label{lem:2dtoda}
For a generic set of $g$ points $\gamma_1, \ldots, \gamma_g$ there is a unique meromorphic function $\Psi_i(t,p), \ p\in \G$ such that: (i) away of the marked points $p_\pm$ it has at most simple poles at $\gamma_s$ (if $\gamma_s$ are distinct); (ii) in the neighborhoods of the marked points it has the form
\beq\label{psipm}
\Psi_i(t,z_\pm)=z_\pm^{\pm i} e^{(\sum_m t_m^{\pm}z_\pm^{-m})}\left(\sum_{s=1}^\infty \xi_s^\pm(i,t)\, \, z_\pm^s\right),\ \ \xi^-_0= 1.
\eeq
\end{lem}
The function $\Psi_i$ is the particular case of , so-called, {\it two-point multi-variable BA function} (see the definition of the multi-point multivariable BA function in \cite{kr-shiot}). Uniqueness of it implies
\begin{theo} \cite{kr-2dtoda} Let $\Psi_i(t,p)$ by the BA function corresponding to any set of data above, i.e.
$\{\G,p_\pm,z_\pm;\gamma_1,\ldots, \gamma_g\, \}$. Then there exist unique operators $L_m^{\pm}$ of the form (\ref{LJ}, \ref{leading}) with $\varphi_i(t):=\ln \xi_0^+(t)$ such that equations (\ref{hierarch}) hold.
\end{theo}

\noindent
{\bf Remark.} By definition the BA function depends on a choice of local coordinates $z_\pm$ in the neighborhood of the marked points $p_\pm$. A change of local coordinate corresponds to triangular transformation of the times $t_m^{\pm}$ (compare with the remark in the Introduction).

\medskip

The algebraic-geometrical solutions of the $2D$ Toda hierarchy can be explicitly written in terms of the Riemann theta-function. Let us choose a basis of cycles $a_i,b_i, \, i=1,\ldots,g$ on $\G$ with the canonical matrix of intersections: $a_i\circ a_j=b_i\circ b_j=0, a_i\circ b_j=\delta_{ij}$. Then one can consequently define:

(a) a basis of the normalized holomorphic differentials $\omega_i,\, \oint_{a_j}\omega_i=\delta_{ij};$

(b) the matrix of their $b$-periods, $B_{ij}=\oint_{b_j}\omega_i$, and the corresponding Riemann theta function
$$\theta(z)=\theta(z|B)=\sum_{m\in \mathbb Z^g} e^{2\pi i (m,z)+\pi i (Bm,m)},\ \ z=z_1,\ldots, z_g;$$

(c) the Abel transform $A(p)$ that is a vector with coordinates $A_k(p)=\int^p \omega_k$;

(d) the normalized Abelian differential of the third kind $d\Omega_{0}, \oint_{a_i}d\Omega_0=0,$ having simple poles with residues $\pm 1$ at $p_{\pm}$, and the normalized abelian differentials of the second kind $d\Omega_{m,\pm}, \oint_{a_i}d\Omega_{m,\pm}=0,$ having poles at $p_\pm$ of the form $d\Omega_{m,\pm}=d(z_{\pm}^{-m}+O(z_{\pm}))$.

\begin{lem} \cite{kr-2dtoda}
The Baker-Akhiezer function is given by the formula
\beq\label{psiformula}
\Psi_i(t,p)=
\frac{\theta(A(p)+iU_0+\sum U_{m,\pm}t_{m}^{\pm}+Z)\, \theta(A(p_-)+Z)}
{\theta(A(p_-)+iU_0+\sum U_{m,\pm}t_{m}^{\pm}+Z)\,\theta(A(p)+Z)}\,
e^{i\Omega_0(p)+\sum t_m^\pm \Omega_{m,\pm}(p)}
\eeq
Here the sum is taken over all pairs of indices $(m,\pm)$ and:

a) $\Omega_0(p)$ and $\Omega_{m,\pm}(p)$ are the abelian integrals,
$
\Omega_0(p)=\int^p d\Omega_0,\ \Omega_{m,\pm}(p)=\int^p d\Omega_{m,\pm},
$
corresponding to the differentials introduced above and normalized by the condition that in the neighborhood of $p_-$ they have the form
$$\Omega_0(z_-)=\ln z_- +O(z_-),\ \Omega_{m,-}(z_-)=z_-^{-m}+O(z_-),\ \Omega_{m,+}(z_-)=O(z_-);$$

b) $2\pi i U_0$, $2\pi iU_{\a,j}$ are the vectors of their $b$-periods, i.e. vectors with the coordinates
\beq\label{uperiods}
U_0^k={1\over 2\pi i} \oint_{b_k} d\Omega_{0},\ U_{m,\pm}^k={1\over 2\pi i} \oint_{b_k} d\Omega_{m,\pm};
\eeq

c)$Z$ is an arbitrary vector (it corresponds to the divisor of poles of
Baker-Akhiezer function).
\end{lem}

Notice, that from the bilinear Riemann relations it follows that $U_0=A(p_-)-A(p_+)$. Then from the comparison of the evaluation of (\ref{psiformula}) one gets:
\begin{theo} \cite{kr-2dtoda} The algebraic-geometrical solutions of the $2D$ Toda lattice are given by the formula
\beq\label{formula}
\varphi_i(t)=\ln \frac{\theta ((i+1)U_0+\sum U_{m,\pm}t_{m}^{\pm}+\widetilde{Z})}{\theta (iU_0+\sum U_{m,\pm}t_{m}^{\pm}+\widetilde{Z})}+ic_0+
\sum c_{m,\pm}t_{m}^{\pm}
\eeq
where $\tilde Z=Z+A(p_+)$ is an arbitrary vector, the vectors $U_0$ and $U_{m,\pm}$ are defined in (\ref{uperiods}),
and the constants $c_0$ and $c_{m,\pm}$ are the leading coefficients of the expansions of the abelian integrals in the neighborhood of $p_+$:
\beq\label{constants}
\Omega_0(z_+)=-\ln z_+ +c_0+O(z_+),
\eeq
$$ \Omega_{m,-}(z_+)=c_{m,-}+O(z_+),\ \Omega_{m,+}(z_+)=z_+^{-m}+c_{m,+}+O(z_+)
$$
\end{theo}
From (\ref{formula}) it is easy to see than in general the algebraic-geometrical solution is {\it quasi-periodic} function of all the variables including $i$. It is $n$-periodic in the discrete variable $i$ if
the vector $nU_0=n(A(p_+)-A(p_-))$ is a vector in the lattice defining the Jacobian of the corresponding curve $\G$. The latter constraint is equivalent to the following:

\begin{lem} \label{lm:per} Let $\G$ be a smooth algebraic curve on which there is a meromorphic function $w$ with the only pole at some point $p_-$, and zero at another point $p_+$ of order $n$ equal to the order of its pole. Then the BA function corresponding to $\G, p_\pm$ and any divisor $\g_s$ satisfies the equation (\ref{bloch}), and therefore the corresponding solutions of $2D$ Toda hierarchy are $n$-periodic.
\end{lem}
For the proof of the statement it is enough to check that the functions $\Psi_{i-n}$ and $w\Psi_n$ have the same analytical properties and hence coincide.

\medskip
\subsection{The dual Baker-Akhiezer function}

For further use, recall also a notion of the dual Baker-Akhiezer function (see details in \cite{kr-shiot}). First for a non-special degree $g$ divisor $D=\g_1+\cdots+\g_g$ on a smooth genus $g$ algebraic curve $\G$ with two marked points one can define a {\it dual} degree $g$ effective divisor $D^+=\g_1^++\cdots+\g_g^+$ as follows: for a given $D$ there exists a unique meromorphic differential $d\Omega$ with simple poles with resides $\pm 1$ at the marked points that is holomorphic everywhere else, and which has zeros at $\g_s$, $d\Omega(\g_s)=0$. The zero divisor of $d\Omega$ is of degree $2g$. Hence, besides of $\g_s$ the differential $d\Omega$ has zeros at $g$ other points $\g_s^+$, i.e. $d\Omega(\g_s^+)=0$. In other words, the divisor $D^+$ is defined by the equation $D+D^+={\mathcal K}+p_++p_-\in J(\G)$ where $\mathcal K$ is the canonical class, i.e. the equivalence class of the zero divisor of a holomorphic differential on $\G$.

If the BA function $\Psi_i(t,p)$ is defined by the divisor $D$ then, its dual function $\Psi_i^+(t,p)$ is uniquely defined by the following analytical properties: (i) away of the marked points $p_\pm$ it is meromorphic and has at most simple poles at $\gamma_s^+$ (if $\gamma_s^+$ are distinct); (ii) in the neighborhoods of the marked points it has the form
\beq\label{psipm+}
\Psi_i^+(t,z_\pm)=z_\pm^{\pm i} e^{-(\sum_m t_m^{\pm}z_\pm^{-m})}\left(\sum_{s=1}^\infty \chi_s^\pm(i,t)\, \, z_\pm^s\right),\ \ \chi^-_0= 1.
\eeq
It is easy to see that the differential $\Psi_i^+\Psi_j d\Omega$ is meromorphic with the only possible poles at $p_\pm$. Moreover, for $i>j$ ($i<j$) it is holomorphic at $p_+$ ($p_-$). Since the sum of residues of a meromorphic differential equals zero, we get the equations
\beq\label{residues}
{\rm res}_{p_\pm} \Psi^+_i\Psi_j\,d\Omega=\pm \delta_{i,j}
\eeq
which imply that $\Psi^+$ satisfies the adjoint equations
\beq\label{adjoint}
(\Psi^{+}L)_{i}\equiv\Psi_{i+k+1}^+ +a_{i+k}^{(k)}\Psi_{k}^+ +\ldots a_{i+1}^{(1)}\Psi_{i+1}^+ =E\psi_{i}
\eeq
and
\beq\label{adjoint1}
-\p_{t_m^{\pm}}\Psi^+=\Psi^+L_m^{\pm}
\eeq
The theta-functional formula (\ref{formula}) for the dual BA function takes the form:
\beq\label{psiformula+}
\psi_i^+(t,p)=
\frac{\theta(A(p)-iU_0-\sum U_{m,\pm}t_{m}^{\pm}+Z^+)\, \theta(A(p_-)+Z^+)}
{\theta(A(p_-)-iU_0-\sum U_{m,\pm}t_{m}^{\pm}+Z^+)\,\theta(A(p)+Z^+)}\,
\times
\eeq
$$\times e^{-i\Omega_0(p)-\sum t_m^\pm \Omega_{m,\pm}(p)}$$
where $Z+Z^+={\mathcal K}+A(p_+)+A(p_-)$.

From the analytical properties of $\Psi^+$ it easy follows that:
\begin{lem}\label{lm:adjper} Under the assumptions of Lemma \ref{lm:per} the dual BA function satisfies the equation
\beq\label{adbloch}
\Psi^+_i=w\Psi^+_{i-n}
\eeq
\end{lem}

\noindent
{\bf Important remark}.
As it was already mentioned above, in the particular case the construction of the algebraic-geometrical solutions of the $2D$ Toda hierarchy can be seen as the construction of the inverse spectral transform. Indeed, let $\G$ be a curve defined by an equation of the form (\ref{spectralcurve}), then a simple comparison of the analytic properties shows that the Bloch function of the operator $L$ coincides with the evaluation of the multivariable BA function at the zero value of all continuous times: $\psi_i=\Psi_i(t_k^\pm=0)$

\section{The Hamiltonian theory of the reduced systems}

The systems in question, namely equations (\ref{LT1}) and (\ref{LT2}) were defined as a special reduction of the $2D$ Toda hierarchy. Therefore, the formulae (\ref{psiformula}), where the Riemann theta-function corresponds to any curve defined by equation(\ref{spectralcurve}) provides solutions to our reduction of the $2D$ Toda hierarchy.  In this section we develop Hamiltonian theory of this reduced system following the general scheme
proposed in \cite{kp1,Kr2}. According to that scheme on {\it the space of operators $L$} which is identified with a phase space of the system, one can define a family of two-forms by the formula
\beq\label{forms}\omega^{(i)}=-\frac{1}{2}\sum_{\a} \res_{p_\a} E^{-i}\langle\psi^+(w)\,\delta L\wedge \delta \psi(w)\rangle \,{d\Omega}
\eeq
where $\delta F(L)$ stands for the differential of a function $F$ on the space of the operators (the BA function with fixed eigenvalue $w$ and fixed normalization is such a function), and the sum is taken over the set of points $p_\a$ on the corresponding spectral curve where the expression in the right hand side a'priory has poles: namely at the marked point $p_{\pm}$, where the BA function and its dual have poles, and for $i>0$ at the zeros $p_\ell, \, \ell=1,\ldots, k$, of the function $E=E(p)$ where $w=w(p)$ does not vanish, i.e. $E(p_\ell)=0, \, w(p_\ell)\neq 0$.

\medskip
\subsection{The differential $d\Omega$}
Our  first goal is to derive a closed expression for $d\Omega$ introduced above by its analytic properties in terms of Bloch eigenfunction of the operator $L$ and its adjoint one.

Let us assume first that the coefficients of the operator are $n$-periodic. Following a line of arguments in \cite{Kr3} consider the differential $d\psi$ with respect to the spectral variable. It satisfies the non homogeneous  linear equation
\beq\label{var}
(L-E)\,d\psi=dE\psi
\eeq
which is just the differential of equation (\ref{eigen}). Taking the differential of equation (\ref{bloch}) we get that $d\psi$ satisfies the following monodromy relation
\beq\label{blochd}
wd\psi_i+dw \psi_i=d\psi_{i-n}
\eeq
For brevity let us denote the average of a function $f_i$ over the interval $l+1\leq i \leq l+n$ by
$\langle f\rangle_l:=\frac 1n \sum_{i=l+1}^{l+n} f_i$ and write $\langle f \rangle $ when that average does not depend on $l$, as in the case of $n$-periodic functions.

From (\ref{var}) it follows that
\beq\label{1}
E\langle\psi^+\,d\psi\rangle_l+dE\,\langle\psi^{+}\psi\rangle=\langle\psi^{+}(Ld\psi)\rangle_l=
\sum_{j=1}^{k+1}\sum_{i=l+1}^{l+n}
a_i^{(j)}\psi_i^+d\psi_{i-j}
\eeq
Equation (\ref{adjoint}) implies
\beq\label{2}
E\langle\psi^+\,d\psi\rangle_l=\sum_{j=1}^{k+1}\sum_{i=l+1}^{l+n}
a_{i+j}^{(j)}\psi_{i+j}^+d\psi_{i}=\sum_{j=1}^{k+1}\sum_{i=l+1+j}^{l+n+j}
a_i^{(j)}\psi_i^+d\psi_{i-j}
\eeq
Subtracting (\ref{2}) from (\ref{1}) and using (\ref{blochd}) we get
\beq\label{3}
dE\,\langle\psi^{+}\psi\rangle=\frac{dw}{nw}\sum_{j=1}^{k+1}\sum_{i=l+1}^{l+j}a_i^{(j)}\psi_i^+\psi_{i-j}
\eeq
Notice that the left hand side of (\ref{3}) does not depend on $l$. Hence, the right hand side of (\ref{3}) is also $l$-independent. Taking the average of the right hand side of (\ref{3}) in $l$ we obtain the equation
\beq\label{decendent}
dE\,\langle\psi^{+}\psi\rangle=\frac{dw}{nw}\langle \psi^+(L^{(1)}\psi)\rangle
\eeq
where
\beq\label{decendent1}
L^{(1)}:=\sum_{j=1}^{k+1}j a_i^{(j\,)}T^{-j}
\eeq
is the difference analog of the first descendent of a differential operator introduced in \cite{Kr3}.

From (\ref{decendent}) it follows that the zeroes of $dw$ coincide with the zeroes of the meromorphic function $\langle\psi^{+}\psi\rangle$ and the zeros of $dE$ coincide with the zeros of $\langle \psi^+(L^{(1)}\psi)\rangle$. Hence:
\begin{lem} The differential
\beq\label{form1}
d\Omega:=\frac{dw}{nw\langle\psi^{+}\psi\rangle}=\frac{dE}{\langle\psi^{+}(L^{(1)}\psi)\rangle}.
\eeq
is holomorphic away of the marked points $p_\pm$, and has zeros at
the poles of $\psi$ and $\psi^+$; at $p_\pm$ it has simple poles with resides $\pm 1$,
i.e $d\Omega$ is the differential introduced above in the definition of the dual BA function.
\end{lem}

\medskip
\noindent
{\bf Example:} $k=1$
\beq\label{d2}
d\Omega=\frac{dE}{\langle a_i^{(1)}\psi_{i}^+\psi_{i-1}^{+}+2\psi_{i}^+\psi_{i-2}\rangle}=\frac{dw}{nw\langle\psi^{+}\psi\rangle}
\eeq

\medskip
\noindent
{\bf Example:} $k=2$
\beq\label{d2}
d\Omega=\frac{dE}{\langle a_i^{(1)}\psi_{i}^+\psi_{i-1}^{+}+2a_i^{(2)}\psi_{i}^+\psi_{i-2}+3\psi_i^+\psi_{i-3}\rangle}=\frac{dw}{nw\langle\psi^{+}\psi\rangle}
\eeq

\subsection {The symplectic leaves and the Darboux coordinates}

It is necessary to emphasize that the form $\omega^{(i)}$ is not closed and is degenerate on the space of {\it all} the operators $L$. It becomes closed after restriction onto certain subvarieties. As we shall see below, only the forms $\omega^{(0)}$ and $\omega^{(1)}$ are non-degenerate on the corresponding subvariety. That allows to regard the total space of operators $L$ as a {\it Poisson} manifold foliated by symplectic leaves of two types. The existence of these two types of foliation reflects, the so-called, bi-Hamiltonian origin of integrable systems.

The constrains defining the symplectic leaves are equivalent to the condition that the form $\omega^{(i)}$ does not depend on a choice of the normalization of the Bloch eigenvector $\psi$. The change of normalization is equivalent to the transformation $\psi_i \to \psi_i h,\ \psi_i^+\to \psi_i^+ h^{-1}$, where $h=h(w)$ is a scalar function. Under this transformation the differential in the right hand side of (\ref{forms}) gets transformed to the differential
\beq\label{form-trans}
E^{-i}\langle\psi^+(w)\,\delta L\wedge \delta \psi(w)\rangle \,{d\Omega}+
E^{-i}\langle\psi^+(w)\,\delta L \psi(w)\rangle \wedge \delta \ln h\,{d\Omega}
\eeq
Hence, the form $\omega^{(i)}$ is {\it normalization} independent only when the last term in (\ref{form-trans}) is holomorphic near the points $p_\a$. From the equation
\beq\label{varE}
(L-E)\delta \psi(w)=-(\delta L-\delta E(w))\psi
\eeq
and the definition of the adjoint operator it follows that
\beq\label{varE1}
\langle \psi^+((\delta L-\delta E)\psi)\rangle=\langle(\psi^+(E-L)) \delta \psi\rangle=0
\eeq
Then using (\ref{form1}) we obtain the following statement:
\begin{lem} \label{leaves}
The restriction of the form  $\omega^{(i)}$ given by (\ref{forms}) onto a subvariety of the space of all operators
such that on  this subvariety the differential $E^{-i} \delta E(w) d\ln w$ is holomorphic in the neighborhoods of points $p_\a$ is normalization independent.
\end{lem}

\medskip
\noindent
{\bf Example $i=0$}. For $i=0$ the sum in (\ref{forms}) is taken over the marked points $p_\pm$ only. At the point $p_+$ (where $w=0$) the function $E$ has zero. Hence, the form $Ed\ln w$ has the only pole at $p_-$. Hence, it has no residue at $p_-$. Therefore for any operator $L$, the corresponding coefficient of (\ref{Eexpansion}) vanishes. Namely, $e_{k+1}\equiv 0$.

In the neighborhood of  $p_-$ where the function $E$ has pole of order $(k+1)$ the form $\delta E(w)d\ln w$ has pole of order $(k+2)$ with no residue.  Hence, if for any set $c=(c_1,\ldots,c_k)$ of constants we define $\Lambda_0^c$ as the subvariety of
operators $L$ satisfying the constraints
\beq\label{leaves0}
\Lambda_0^c:=\{L\in \Lambda_0^c| \,e_s(L)=c_s, s=1,\ldots k\}\,
\eeq
where $e_s=e_s(L)$ are the coefficients of expansion (\ref{Eexpansion}), then:
\begin{cor}
The form $\omega^{(0)}$ restricted to the subvariety $\Lambda_0^c$
is normalization independent.
\end{cor}

\medskip
\noindent
{\bf Example $i=1$}. The $E^{-1}\delta E(w) d\ln w$ is holomorphic at the marked point $p_-$. Since,
the sum of its residues equals zero, it is holomorphic at the point $p_+$ if it is holomorphic at the points $p_\ell, \, \ell=1\ldots, k$. Using the chain rule we get that the variation of $E(w)$ with fixed $w$ is related with the variation of $w(E)$ with fixed $E$ by the formula $\delta E(w)\,dw+\delta w(E)\,dE=0$. Hence, $\delta \ln E(w)\,d\ln w$ is holomorphic at $p_\ell$ (the preimages of $E=0$ where $w\neq 0$) if the equations $\delta w(p_\ell)=0$ are satisfied. The latter hold along the subvariety
\beq\label{leaves1}
\Lambda_1^c:=\{L\in \Lambda_1^c| \,r_{i,\,0}(L)=c_i, 1=1,\ldots k\}
\eeq
where $c=(c_1,\ldots,c_k)$ are constants and $r_{i,\,0}(L)=r_{i,\,0}$ are the coefficients of the polynomial $\det L(w)=w^{k+1}+\sum_{i=1}^{k} r_{i,\,0}w^i$.
\begin{cor} \label{lm:constraint}
The form $\omega^{(1)}$ restricted to the subvariety $\Lambda_1^c$
is normalization independent.
\end{cor}

\medskip
\begin{rem} For $i>1$ the subvarieties $\Lambda^c_i$, on which the restriction of $\omega^{(i)}$ is normalization independent, are described in a similar way by a system of $i(k+1)-1$ equations:
\beq\label{leavesI}
\Lambda_i^c:=\{L\in \Lambda_i^c| \, w_{\ell,s}=c_{\ell,s}, s=1,\ldots, i;\, w_s=c_s, s=2,\ldots, i \}
\eeq
where $w_{\ell,s}$ are the coefficients of the expansions
\beq\label{wjexpansion}
w=\sum_{s=0}^\infty w_{\ell,s}E^s
\eeq
of $w$ at the preimages $p_\ell$ on $\G$ of $E=0$ at which $w(p_\ell)\neq 0$; $w_s$ are the coefficients of the expansion (\ref{wzero}) of $w$ at $p_+$  and $c_{i,s}, c_s$ are constants. Hence $\Lambda_i^c$ is of dimension $(n-1)k-i+1$. Recall that the family of curves $\G$ defined by  equations of the form (\ref{spectralcurve}) is of dimension $\frac{k(n+1)}2$ (the number of the coefficients $r_{ij}$). For generic values of coefficients $r_{ij}$ the curve $\G$ is smooth and has genus $g=\frac{k(n-1)}2 $\,. Therefore, the correspondence (\ref{correspondence}) restricted to $\Lambda_i^c$ identifies the latter with the total space of the Jacobian bundle over the space of the corresponding spectral curves. For $i>1$ the dimension of the fiber is {\it bigger} then the dimension of the base. Hence the form $\omega^{(i)}$ restricted to $\Lambda_i^c$ is degenerate for $i>1$.
\end{rem}
\subsection{The Darboux coordinates}
For completeness, let us present a construction of the Darboux coordinates for the restriction $\wh\omega^{(i)}$ of $\omega^{(i)}$ onto the subvariety $\Lambda_i^c$, i.e.
\beq\label{omegarestriction}
\wh\omega^{(i)}:=\omega^{(i)}|_{\Lambda^c_i}
\eeq

\begin{theo}
Let $\g_s$ be the poles of the BA function. Then the equation
\beq\label{d1}
\wh\omega^{(i)}=\frac1n\sum_{s=1}^{g} E^{-i}(\g_s) \delta E(\g_s)\wedge \delta\ln w (\g_s).
\eeq
holds.
\end{theo}
\noindent
The meaning of the right hand side of this formula is as follows.
The spectral curve is equipped by definition
with the meromorphic functions $E$ and $w$.  The evaluations $E(\g_s),\ w(\g_s)$ at the points
$\g_s$ define functions on the space of $L$ operators. The wedge product of
their external differentials is a two-form on our phase space.

\bigskip
\noindent{\it Proof.}
The proof of the formula (\ref{d1}) is very general
and does not rely on any specific form of $L$. Let us present it briefly
following the proof of Lemma 5.1 in \cite{kr4} (more details can be found in
\cite{spin}).
The differential whose residues define $\omega^{(i)}$ by (\ref{forms}) is a meromorphic differential on the
spectral curve $\G$. Therefore, the sum of its residues at the punctures
$p_\a$  is equal to the negative of the sum of the other
residues on $\G$. There are poles of two types. First of all, the differential has poles at the poles
$\g_s$ of $\psi$. Note that $\delta \psi$ has pole of the second order
at $\g_s$. Taking into account that $d\Omega$ has zero at $\g_s$
we obtain
\beq\label{65}
\res_{\g_s} E^{-i}\langle\psi^+\,\delta L\wedge \delta \psi\rangle \,{d\Omega}=
\frac{E^{-i}\langle\psi^+\delta L\psi\rangle}{n\langle\psi^+\psi\rangle}(\g_s)
\wedge \delta \ln w(\g_s)=
\eeq
$$
=\frac 1nE^{-i}(\g_s)\delta E(\g_s)\wedge\delta \ln w(\g_s).
$$
The last equality follows from (\ref{varE1}) which is just
the standard formula for the variation of an eigenvalue of an  operator.

The second set of poles of the differential in the righthand side of (\ref{forms}}) is the set of zeros $q_j$ of the differential $dw$. Indeed, in the neighborhood of $q_j$ the local coordinate on the spectral curve is $\sqrt{w-w(q_j)}$ (in general position when the zero is simple).
Taking a variation of the Taylor expansion of $\psi$ in that coordinate
we get that
\beq
\delta \psi=-{d\psi\over dw}\delta w(q_j)+O(1).\label{66}
\eeq
Therefore, $\delta \psi$ has simple pole at $q_j$. In the similar way
we have
\beq
\delta E=-{dE\over dw} \delta w(q_j). \label{67}
\eeq
Equalities (\ref{66}) and (\ref{67}) imply that
\beq
\res\,_{q_j}
E^{-i}\langle\psi^{+} \delta L\wedge \delta \psi\rangle d\Omega=
\res_{q_j}\frac{E^{-i}\langle\psi^+\delta L d\psi\rangle}{n\langle\psi^+\psi\rangle} \wedge
{\delta E d\ln w\over
dE} \label{68}
\eeq
Due to skew-symmetry of the wedge product we we may replace $\delta L$ in
(\ref{68}) by $(\delta L-\delta E)$. Then, using the identities
$\psi^*(\delta L-\delta E)= \delta \psi^* (E-L)$  and
$(E-L)d\psi=-dE\psi$, one gets
\beq\label{000}
\res_{q_j}
E^{-i}\langle\psi^{+} \delta L\wedge \delta \psi\rangle d\Omega
=-\res_{q_j}\frac{E^{-i}\langle\delta \psi^+\psi\rangle}{n\langle\psi^+\psi\rangle}\wedge \delta E d\ln w=
\eeq
$$
=\res_{q_j}\frac{E^{-i}\langle\psi^+\delta \psi\rangle}{n\langle\psi^+\psi\rangle}\wedge \delta E d\ln w\,,
$$
where in the last equality we use the identity $\langle\psi^+\psi\rangle(q_j)=0$ (which follows, as we already stressed, from (\ref{decendent})). By definition of the locus on which $\omega^{(i)}$ is normalization independent (see Lemma \ref{leaves}) the form in the right hand side of (\ref{000}) has no poles at the points $p_\a$. Besides of poles at $q_i$ it has poles at $\g_s$, only. Hence, after the restriction on the leave we get the equation

\beq\label{f000}
\sum_j\res_{q_j}\frac{E^{-i}\langle\psi^+\delta \psi\rangle}{n\langle\psi^+\psi\rangle}\wedge \delta E d\ln w=
-\sum_s\res_{\g_s}\frac{E^{-i}\langle\psi^+\delta \psi\rangle}{n\langle\psi^+\psi\rangle}\wedge \delta E d\ln w=
\eeq
$$=\frac 1n\sum_s  E^{-i}(\g_s)\delta E(\gamma_s) \wedge \delta \ln w(\g_s).$$

Equations (\ref{65},\ref{000}, \ref{f000}) directly imply (\ref{d1}). The theorem is proved.

\subsection {The Hamiltonians}

The next step in the construction of the Hamiltonian theory for the systems admitting the Lax representation is to show that the contraction of the form $\omega^{(i)}$, restricted to the subvariety, where it is normalization independent, with the vector field $\p_t$ defined by the Lax equation is an exact one form, i.e. $\widehat\omega^{(i)}(\p_t,X)=\delta H^{(i)}(X)$. Then, on any subvariety on which the form $\widehat\omega^{(i)}$ is non-degenerate the vector-field $\p_t$ is Hamiltonian with the Hamiltonian $H$.

Below we apply the general scheme to equations (\ref{LT1}) and (\ref{LT2}) and explicitly compute the corresponding Hamiltonians. Let $\p_t$ be the vector field defined by the Lax equations, then
\beq\label{lax-forms}
\p_t\, L=[M,L],\ \ \p_t \psi=M\psi-\psi f
\eeq
where $f$ is a meromorphic function on the spectral curve.

\medskip
\noindent
\begin{rem}
The appearance of the term with $f$ in the expression for $\p_t \psi$ is due to the fact that in the definition of the $\omega^{(i)}$ it is assumed that the normalization of the Bloch function $\psi$ is {\it time-independent}: $\psi_0\equiv 1$.
With that normalization if the operator $L$ depends on $t$ according to a Lax equation, then the spectral curve $\G$ is time independent and the time dependence of the pole divisor $D(t)$ of $\psi(t)$ becomes linear after the Abel transform. The latter
follows form the relation
\beq\label{time-normalization}
\psi_i(t,p)=\Psi_i(t,p)\Psi_0^{-1}(t,p)
\eeq
where $\Psi$ is the multi variable BA function given by  (\ref{psiformula}). Then the equation (\ref{hierarch}) implies
equation (\ref{lax-forms}) with $f(t,p)=\p_t\ln \Psi_0(t,p)$. The function
$f$ has poles at the marked points $p_\pm$ of the form
\beq\label{f-infty}
f=\sum_{s=1}^{m_\pm} c^{\pm}_s z^{-s}+O(1)
\eeq
where $c_s^{\pm}$ are {\it constants} which in fact parameterize commuting flows of the hierarchy and $m_\pm$ are positive and negative orders of the operator $M$.
\end{rem}

\begin{theo}\label{thm:1} The vector-field $\p_{t_m^{\pm}}$ defined by Lax equation (\ref{lax}) restricted to the subvariety $\Lambda^c_i$ is Hamiltonian for $i=0,1$ with respect to the form $\wh\omega^{(i)}$, and with the Hamiltonian
\beq\label{H-}
H_{t_m^{-}}^{(0)}=\res_{p_-} z^{-m}E(z)d\ln z=e_{m+k+1}
\eeq
\beq \label{H-1}
H_{t_m^{-}}^{(1)}=\res_{p_-} z^{-m}\ln E(z)d\ln z
\eeq
where $E(z)$ is the series (\ref{Eexpansion}) with the coefficients defined in Lemma \ref{L1}, and
\beq\label{H+}
H_{t_m^{+}}^{(i)}=\frac 1n \res_{p_+} E^{-m-i}\ln w(E) dE,\ \ i=0,1
\eeq
where $w(E)$ is defined in (\ref{wzero})
\end{theo}
\noindent
{\it Proof.} The substitution of (\ref{lax-forms}) and (\ref{form1} into (\ref{forms}) gives
\beq\label{subst1}
\omega^{(i)}(\p_t,\cdot)=-\frac12 \sum_{p_\a} \res_{p_\a}\left( \langle\psi^+ [M,L]\delta \psi\rangle-
\langle\psi^+\delta L(M\psi-\psi f)\rangle\right)\frac {d\ln w}{nE^{i}\langle \psi^+\psi\rangle}
\eeq

Using once again the equation $(L-E)\delta \psi=-(\delta L-\delta E)\psi$ we get
that the differential in the right hand side of (\ref{subst1}) is equal to
\beq\label{subst2}
-\frac 12 \left(\langle \psi^+(M\delta E+\delta L f) \psi\rangle-\langle \psi^+(\delta L\,M+M\,\delta L)\psi\rangle
\right)\frac{d\ln w}{nE^{i}\langle \psi^+\psi\rangle}
\eeq
The second term has poles only at the points $p_\a$. Hence, the sum of its residues at these points is equal to zero.
The first term is equal to
\beq\label{subst3}
-\frac12 \langle \psi^+(2f+(M-f))\psi\rangle \delta E\frac{d\ln w}{nE^{i}\langle \psi^+\psi\rangle}
\eeq
From the  definition of $f$ in (\ref{lax-forms}) it follows that $\langle\psi^+(M-f)\psi\rangle$ is holomorphic at $p_\a$.
Since the restriction of $E^{-i}\delta E d\ln w$ onto $\Lambda_i^c$ is holomorphic at the marked points $p_\a$, the second term in (\ref{subst3}) restricted to $\Lambda_i^c$ has no residues at $p_\a$. Recall that the function $f$ has poles only at the points $p_\pm$. Using the identity $\delta E(w)d\ln w=-\delta \ln w(E)dE$ for the residue at $p_+$ we finally obtain the equation
\beq\label{Hfinal}
\widehat \omega^{(i)}(\p_t,\cdot)=\frac 1n \res_{p_+}f(E)\delta \ln w(E) E^{-i}dE-\frac 1n\res_{p_-} f(w) E^{-i}(w)\delta E(w) d\ln w
\eeq
Recall, that the choice of the basis vector fields $\p_{t_m^{\pm}}$ of the hierarchy is fixed by the choice of local  coordinates near the marked point $p_{\pm}$. As it follows form the constructions of Lemma \ref{L1} and Lemma \ref{L2} the most natural choice is $z=w^{-1/n}$ at $p_-$ and $z=E$ at $p_+$. With this choice of local the function $f_m^{\pm}$ corresponding to $t=t_m^{\pm}$ has pole at $p_{\pm}$ of the form $f_m^+=E^{-m}+O(E)$ and $f^{-}_m= z^{-m}+O(z), z=w^{-1/n}$, respectively. Then
(\ref{Hfinal}) implies $\wh \omega^{(i)}(\p_{t_m^\pm},\cdot)=\delta H^{(i)}_{t_m^\pm}$. The theorem is proved.

\section{Special coordinate systems. Examples}

We begin this section by presenting certain systems of coordinates on the space of lower triangular operators in which the forms $\omega^{(\ell)}, \ell=1,2,$ have {\it local densities}. By latter we mean coordinates $x_i^{(j)}$ in which the form can be written as $\omega=\sum f_{i,i_1}^{(j,j_1)}\delta x_s^{(j)}\wedge \delta x_{i_1}^{(j_1)}$ where the sum is taken over the set of all the indices with $|i-i_1|<d_1$ for some integer $d_1$ independent on the period $n$ of the operator. Moreover it is assumed also that the coefficients $f_{i,i_1}^{(j,j_1)}$ are functions of $x_{i_2}^{(j_2)}$ such that $|i-i_2|< \ell_2$ for some $n$-independent integer $\ell_2$.

\begin{rem}Note, that in the original coordinates on the space of lower triangular operators that are the coefficients $a_i^{(j)}$ of the operators are not of the type we are looking for. Indeed, by definition the densities of the forms involve the coefficients of the expansions of $\psi_i$ at the points $p_\a$ which are not local.
\end{rem}
\medskip

\subsection {The form $\omega^{(0)}$}
The first coordinate system, in which the form $\omega^{(0)}$ has local density, we identify with the set of of the first $k$ coefficients of the expansion (\ref{formalpsi}) of the Bloch solution at the marked point $p_-$. The fomulae (\ref{xi1}, \ref{xi})
for $s=1,\ldots k$ can be regarded as the definition of the map
\beq\label{map-local}
\{\xi_s^-(i),e_s\}\longmapsto \{a_i^{(j)}\}
\eeq
where the variables functions $\xi_s^-(i)$ are defined up to a common shift $\xi_s^-(i)\to \xi_s^-(i)+c_i$ or equivalently normalized by the constraint $\xi_s^-(0)=0$.

The form $\omega^{(0)}$ by definition in (\ref{forms}) is an average over $i$ of some expressions involving $\xi_s^-(i-j), j=0\ldots k$,  and the first $(k-1)$ coefficients of the expansion at $p_-$ of the function
\beq\label{psi*}
\psi_i^*:=\frac {\psi_i^+}{\langle \psi^+\psi\rangle}
\eeq
where $\psi^+$ is the dual BA function (\ref{psipm+}). The coefficients of $\psi^*_i$ can be found recurrently from the relations
\beq\label{residues1}
\res_{p_-} \psi_i^*\psi_{i-j}d\ln z=\delta_{0,j}
\eeq
which follow from (\ref{residues}) and (\ref{form1}). Hence, the expression of any of these coefficients  in terms of $\xi_{s}^-$ is local. Then the statement that $\omega^{(0)}$ has local density in the new coordinates is an obvious corollary of the definition.

\medskip
\noindent{\bf Example $k=1$} The initial coordinates on the space of $n$ periodic lower triangular operators of order are their coefficients $a_i$: $L=a_iT^{-1}+T^{-2}$. The new coordinates are $x_i:=\xi_1^-(i)$  defined up to a common shift and a constant $e_1$.
The expression for old coordinates in terms of new ones is given by formula (\ref{xi1}):
\beq\label{xitoa}
a_i=x_i-x_{i-2}+e_1
\eeq
The substitution of the expansion of $\psi$ and $\psi^+$ into (\ref{forms}) gives for $k=1$ the following expression for the restriction of $\omega^{(0)}$ onto the symplectic leaf $e_1=const$:
\beq\label{form01}
\wh\omega^{(0)}=\frac 12 \langle d a_i\wedge dx_{i-1}\rangle=\langle dx_i\wedge dx_{i-1}\rangle
\eeq
where as before $ \langle \cdot\rangle$ denotes the average of a periodic  expression in brackets over the period.

\begin{rem} In the formulae above the differential on the phase space (the space of parameters) was denoted by $\delta$ in order do distinguish it form the differential $d$ with respect to the spectral parameter. After taking the residues of the differential in the spectral parameter, here and below we change $\delta$ in our notations to more conventional one, i.e. $dx_i:=\delta x_i$.
\end{rem}

\medskip
According to Theorem \ref{thm:1}, equations (\ref{lttoda1}) restricted to the symplectic leaf $\langle a_i\rangle=\langle e^{\varphi_i-\varphi_{i-1}}\rangle=e_1=const$ are Hamiltonin with respect $\wh \omega^{(0)}$ with the Hamiltonian $H^{(0)}_{t_1^-}:=e_3$. In order to find its explicit expression in term of the new coordinates we use the equations (\ref{xi}). For $s=2, k=1$ we have
\beq\label{x2}
\xi_2^-(i)-\xi_2^-(i-2)+e_1\xi_1(i)^-+e_2=a_i\xi_1^-(i-1)
\eeq
Then from (\ref{xitoa}) it follows
\beq\label{x21}
\xi_2^-(i)-\xi_2^-(i-2)+e_2=x_ix_{i-1}-x_{i-1}x_{i-2}+e_1(x_{i-1}-x_i).
\eeq
Taking the average of equation (\ref{x21}) we get $e_2=0$ (recall that in the proof of Lemma \ref{leaves} it was shown that $e_{k+1}=0$ for any $k$). For $s=3,k=1$ equation (\ref{xi}) has the form
\beq\label{x3}
\xi_3^-(i)-\xi_3^-(i-2)+e_1\xi_2^-(i)+e_3=a_i\xi_2^-(i-1)=(x_i-x_{i-2}+e_1)\xi_2^-(i-1)
\eeq
Taking the average of (\ref{x3}) we get the explicit expression for the Hamiltonian of equation (\ref{lttoda1}) in terms of the new coordinates:
\beq\label{e3}
H^{(0)}_{\p_{t_1^-}}=e_3=\langle(x_i-x_{i-2})\xi_2^-(i-1)\rangle=\langle x_i(\xi_2^-(i-1)-\xi_2^-(i+1)\rangle=
\eeq
$$
=\langle x_i^2(x_{i-1}-x_{i+1})\rangle
$$
where for the last equation we use (\ref{x2}).

\medskip
\noindent
{\bf Example:} k=2. The expressions of the coefficients of a lower triangular operator of order $3$ in terms of the coordinates $x_i:=\xi_1^-(i)$ and $y_i:=\xi_2^-(i)$ are as follow:
\beq\label{k31}
a_i^{(2)}=x_i-x_{i-3}+e_1
\eeq
\beq\label{k32}
a_i^{(1)}=y_i-y_{i-3}+e_1x_i+e_2-a_i^{(2)}x_{i-2}=
\eeq
$$
y_i-y_{i-3}-(x_i-x_{i-3})x_{i-2}+e_1(x_i-x_{i-2})+e_2
$$
The substitution of expansions for $\psi$ and $\psi^+$ into (\ref{forms}) gives the following
\beq\label{omega1}
\omega^{(0)}=\frac 12 \langle\, da_i^{(1)} \wedge dx_{i-1}+da_i^{(2)}\wedge (\chi_1^-(i)\,dx_{i-2}+
d\xi_2^-(i-2))\,\rangle
\eeq
where $\chi_1^-$ is the first coefficient of the expansion of $\psi^+$ at the marked point $p_-$.
Equation (\ref{residues1}) with $j=1$ implies $\chi_1^-(i)=-x_{i-1}$. Then after relatively long but
straightforward computations we obtain the following expression for $\omega^{(0)}$ restricted to a leaf
along which $e_1$ and $e_2$ are constants:
\beq\label{omega1f}
\wh\omega^{(0)}=\langle\,dy_i\wedge (dx_{i-1}-dx_{i+2})
+d(x_{i-1}x_{i-2})\wedge dx_i  \,\rangle
\eeq
$$+e_1\langle dx_i\wedge dx_{i-1}\rangle
$$
Equation (\ref{lttoda1}) for $k=2$ restricted to a leaf with $e_1$ $e_2$ being some constants is Hamiltonian with respect to the form (\ref{omega1f}) and with the Hamiltonian $H^{(0)}_{t_1^-}=e_4$. Similar to  computations of $e_3$ above we get the following expression for the Hamiltonian $H:=e_4$:
\beq\label{e4}
H=\langle\, y_{i-1}(y_i-y_{i-3})>+<x_ix_{i-1}x_{i-2}(x_{i-1}-x_i)>+e_1<(x_{i}^2(x_{i-1}-x_{i+1})\rangle+
\eeq
$$
+e_2\langle x_{i-1}(x_i-x_{i-1})\rangle+\langle y_i(x_{i+2}^2-x_{i-1}^2-x_{i+2}x_{i+1}+x_{i-2}x_{i-1}\rangle
$$
\medskip

\subsection{The form $\omega^{(1)}$}
A system of coordinates in which the form $\omega^{(1)}$ becomes local is suggested by its definition in (\ref{forms}) which involves the evaluation of $\psi_i$ at the marked points $p_\ell\in \G$ that are the preimages of $E=0$ with $w(p_\ell)\neq 0$.

Let $\varPhi=\{\phi_i^\ell\}$ be a $(k\times n)$ matrix of rank $k$, i.e. $i=1,\ldots,n; \ell=1\ldots,k$ . We call two matrices equivalent $\varPhi\sim \varPhi'$ if $\varPhi'=\varPhi \lambda$, where $\lambda={\rm diag}(\lambda_1,\ldots,\lambda_k)$. The space of equivalence classes $[\varPhi]:=(\varPhi/\sim)$ can be seen as the space of {\it ordered sets} of $k$ distinct points in $n-1$ dimensional projective space, $ [\phi^\ell]\in \mathbb P^{n-1}$.

Consider the space of pairs $\{[\,\varPhi],W\}$ where $W=\{w_1,\ldots,w_k\}$ is a set of non-zero numbers, $w_\ell\neq 0$. The symmetric group $S_k$ acts on the space of such pairs by simultaneous permutation of rows of matrix $\varPhi$ and coordinates of the vector $W$. Now we are going to define a map from the corresponding factor-space to the space of $n$-periodic
operators $L$ of the form (\ref{operator}):
\beq\label{map}
\{[\,\varPhi],W\}/S_k\longmapsto L
\eeq
First note that given a set $W=\{w_1,\ldots,w_k\}$ of non-zero numbers any $(k\times n)$ matrix $\varPhi$ can be extended
to a unique $(k\times \infty)$ matrix $\phi_i^\ell, i\in \mathbb Z,$ such that the equation
$\phi_{i-n}^\ell=w_\ell\phi_i^\ell$ holds. Then it is easy to see that there is a unique operator $L$ of the form (\ref{operator}) such that for any $\ell$ the sequence $\phi^\ell=\{\phi_i^\ell\}$ is a solution of the equation
\beq\label{zerolevel}
L\phi^\ell=0 \ \Leftrightarrow \ \sum_{j=1}^{k} a_i^{(j)}\phi_{i-j}^\ell=-\phi_{i-k-1}^\ell
\eeq
Indeed, for fixed $i$ the system (\ref{zerolevel}) is a system of $k$ non-homogeneous linear equations for unknown coefficients
of $L$. Hence, by Cramer's rule:
\beq\label{cramersA}
a_i^{(j)}=-\frac {|\phi_{i-1},\ldots,\phi_{i-j+1},\phi_{i-k-1},\phi_{i-j-1},\ldots,\phi_{i-k}|}
{|\phi_{i-1},\ldots,\phi_{i-j+1},\,\phi_{i-j}\ ,\,\phi_{i-j-1},\ldots,\phi_{i-k}|}
\eeq
Here and below we use the following notations: $\phi_i$ is the $k$-dimensional vector with coordinates $\phi_i:=\{\phi_i^\ell\}$, and for any set $V_1, \ldots , V_k$ of $k$-dimensional vectors $|V_1,\ldots, V_k|$ stands for the determinant of the corresponding matrix, i.e. $|V_1,\ldots, V_k|:=\det \left(V_i^\ell\right)$.

Recall that throughout the paper we use parametrization of the leading coefficient  $a_i^{(1)}$ by variables $\varphi_i$ such that
$a_i^{(1)}=e^{\varphi_i-\varphi_{i-1}}$. Equation (\ref{cramersA}) for $j=1$ allows to identify these variables with
\beq\label{notX}
e^{-\varphi_i}:=(-1)^{ik}|\phi_{i-1},\ldots,\ldots,\phi_{i-k}|
\eeq
Then for any $j$ equation (\ref{cramersA}) takes the form
\beq\label{cramerA1}
a_i^{(j)}=(-1)^{ik+1} e^{\varphi_i}|\phi_{i-1},\ldots,\phi_{i-j+1},\phi_{i-k-1},\phi_{i-j-1},\ldots,\phi_{i-k}|
\eeq

\begin{theo} The map (\ref{map}) defined by formulae (\ref{notX}, \ref{cramerA1})
is one-to-one correspondence between open domains. Under this correspondence
equation (\ref{LT1}) and (\ref{LT2}), restricted to leaves with $w_\ell$ being fixed constants, are Hamiltonian with respect to the form
\beq\label{formfinal}
\wh \omega^{(1)}= \frac 12 \langle d\varphi_{i-1}\wedge d\varphi_i
- (-1)^{(i-1)k}e^{\varphi_{i-1}}\sum_{j=1}^k
 d a_i^{(j)}\wedge
|\phi_{i-2},\ldots,\phi_{i-k},d\phi_{i-j}|
\rangle
\eeq
and with the Hamiltonians
\beq\label{Ham-xi-eta}
H^-=\langle a_i^{(k)}\rangle\,; \ \ H^+=-\langle a_i^{(2)}e^{\varphi_{i-2}-\varphi_i}\rangle
\eeq
respectively.
\end{theo}
\noindent
{\it Proof.} The right hand side of (\ref{cramersA}) is symmetric with respect to the simultaneous permutation of rows
of the matrices in the numerator and denominator. Hence, the map (\ref{map}) is well-defined on an open domain when all denominators are non zero. The inverse map is defined by identification of $w_\ell$ with non zero roots of the polynomial
$R(w,0)=\det L(w)$ defined in (\ref{specR}). In other word $w_\ell$ is an evaluation of the function $w(p)$ on the spectral curve $\G$ of $L$ at one of the preimages on $\G$ of $E=0$, i.e. $p_\ell:(w_\ell,0)\in \G$. Under this identification it easy to see that $\phi_i$ is just the evaluation of the BA function at $p_\ell$, i.e. $\phi_i^\ell=\psi_i(p_\ell)$. Hence the first statement of the theorem is proved.

Recall that by definition, $\omega^{(1)}$ is equal to the average over $i$ of the sum of resides at
$p_{\pm}$ and $p_\ell$ of the form
\beq\label{form12}
-\frac 1{2n}\sum_{j=1}^k\delta a_i^{(j)} \wedge \left(\psi_i^*\delta \psi_{i-j}\right) E^{-1}d\ln w
\eeq
The BA function $\psi_i$ and its dual $\psi_i^+$ has zero and pole of order $i$ at $p_-$, respectively. Since,  $E$ at $p_-$
has pole of order $k+1$, the form (\ref{form12}) is holomorphic at $p_-$. Hence, it has no residue at $p_-$. At $p_+$ the function
$E$ has simple zero. Therefore, the form $E^{-1}d\ln w$ at $p_+$ has pole of order 2. At the same time at $p_+$ the functions
$\psi_i^+$ and $\psi_i$ have zero and pole of order $i$, respectively. Hence, the terms in sum ({\ref{form12}) with $j>1$
are holomorphic at $p_+$. From (\ref{two}, \ref{psipm+}) it follows that

\beq\label{res1}
-\frac 1{2n}\res_{p_+} \delta a_i^{(1)} \wedge \left(\psi_i^*\delta \psi_{i-1}\right) E^{-1}d\ln w=
-\frac 12 \delta (e^{\varphi_i-\varphi_{i-1}})\wedge e^{-\varphi_i} \delta (e^{\varphi_{i-1}})
\eeq
$$=\frac 12 \delta \varphi_{i-1}\wedge \delta \varphi_i
$$
Our next goal is to express $\psi_i^{+}(p_{\ell})$ in terms of  $\phi^\ell=\psi(p_\ell)$, which then will allow us to obtain a closed expression of $\omega^{(1)}$ in terms of $\phi^\ell$.

\begin{lem} Let $r_\ell$ be constants equal $r_\ell:=\res_{p_\ell} E^{-1}d\Omega$. Then
\beq\label{psikr}
r_\ell\psi_{i}^{+}(p_{\ell})=\frac{(-1)^{\ell+k-1} \det \wh \Phi_i^{\ell,k}}{|\phi_{i-2},\ldots,\phi_{i-k-1}|}
\eeq
where $\wh \Phi_i$ is $(k\times k)$ matrix with columns  $(\phi_{i-1},\ldots, \phi_{i-k})$, and $\wh \Phi_i^{\ell,k}$ is
obtained from $\wh \Phi_i$ by removing $\ell$-th row and the last column.
\end{lem}

\begin{proof} By definition of $d\Omega$ the differential $\psi_{i}^{+}\psi_{i-j}E^{-1}d\Omega$ is holomorphic away of the marked points $p_\pm$ and the points $p_\ell$ where $E$ vanishes. For $2\leq j\leq k$ it is holomorphic at $p_\pm$. Hence the sum of its residues at $p_\ell$ equals zero:
\beq\label{25}
\sum_{\ell=1}^{k} \res_{p_\ell}  \ \psi_{i}^{+}\psi_{i-j}E^{-1}d\Omega=\sum_\ell r_\ell \psi_i^{+}(p_\ell)\phi_{i-j}^\ell=0, \ j=2,\ldots,k
\eeq
The differential $\psi_{i}^{+}\psi_{i-j}E^{-1}d\Omega$ is holomorphic at $p_-$ and has simple pole at $p_+$ with residue $-1$.
Hence,
\beq\label{26}
\sum_{l} \ \res_{p_\ell} \psi_{i}^{+}\psi_{i-k-1}E^{-1}d\Omega=\sum_\ell r_\ell \psi_i^{+}(p_\ell)\phi_{i-k-q}^\ell=1.
\eeq
Equations (\ref{25}, \ref{26}) is a system of linear equations for unknown $r_\ell \psi_i(p_\ell)$. Then Cramer's rule implies (\ref{psikr}).
\end{proof}
Note, that multiplying the right hand side of (\ref{psikr}) by $d \phi_{i-j}^\ell$ and then taking the sum over $\ell$ we can identify the latter with an expansion of the  determinant below with respect to the last column, i.e.
\beq\label{27}
-\frac 12 \sum_{\ell=1}^k r_\ell\psi_{i}^{+}(p_{\ell}) d\phi^{\ell}_{i-j}=-\frac12 \frac{|\phi_{i-2},\ldots,\phi_{i-k}, d\phi_{i-j}|}
{|\phi_{i-2},\ldots,\phi_{i-k-1}|}=
\eeq
$$
=\frac{(-1)^{k(i-1)+1}}2 \, |\phi_{i-2},\ldots,\phi_{i-k}, d\phi_{i-j}| e^{\varphi_{i-1}}
$$
The right hand side of (\ref{formfinal}) is equal to the sum of (\ref{res1}) and wedge product of (\ref{27}) with $da_i^{(j)}$.
In order to complete the proof of the theorem it remains only to note,that according to Theorem \ref{thm:1} the Hamiltonians of equations (\ref{LT1}) and (\ref{LT2}) are equal
\beq\label{Hminus}
H^-:=H_{\p_{t_1^-}}=\res_{z=0} \ln E(z)z^{-2}dz= e_1=\langle a_i^{(k)}\rangle
\eeq
and
\beq\label{Hplus}
H^+:=H_{\p_{t_1^+}}=\frac1n\res_{E=0} \ln w(E)E^{-2}dE=w_1
\eeq
where $w_1$ are the first coefficient of expansion (\ref{wzero}). Note, that according to Corollary \ref{cor:1}
\beq\label{lnw}
n^{-1}\ln w=n^{-1}(\ln \psi_{-n}-\ln\psi_0)=\langle \psi_{i-1}-\psi_i \rangle
\eeq
Then, from (\ref{two}) and (\ref{chi1}) we obtain
\beq\label{w1}
w_1=\langle \xi_1^+(i-1)-\xi_1^+(i)\rangle=-\langle a_i^{(2)}e^{\varphi_{i-2}-\varphi_i}\rangle
\eeq
and the Theorem is proved.

\bigskip
\noindent
{\bf Example} For $k=1$ equation (\ref{notX}) takes the form $e^{-\varphi_i}=(-1)^{i}\phi_{i-1}$. Then
\beq\label{finalk1}
\omega^{(1)}=\frac12 \langle d\varphi_{i-1}\wedge d\varphi_i - (-1)^{i-1}e^{\varphi_{i-1}}d (e^{\varphi_i-\varphi_{i-1}})\wedge d\phi_{i-1}\rangle=
\langle d\varphi_{i-1}\wedge d\varphi_i \rangle
\eeq
Note, also that for $k=1$ the coefficient $a_i^{(2)}=1$ and formulae (\ref{Ham-xi-eta}) takes the form (\ref{ham})

\end{document}